\documentclass[%
 aip,
 amsmath,amssymb,
 reprint,%
]{revtex4-1}

\usepackage{graphicx}% Include figure files
\usepackage{dcolumn}% Align table columns on decimal point
\usepackage{bm}% bold math
\usepackage[utf8]{inputenc}
\usepackage[T1]{fontenc}
\usepackage{mathptmx}
\usepackage{physics}
\usepackage{amsthm}
\usepackage{amsfonts}
\usepackage{hyperref} 
\usepackage{cleveref} 
\usepackage{color}

\newtheorem{theorem}{Theorem}[section]
\newenvironment{thm}{\begin{theorem}}{\hfill$\diamond$\end{theorem}}
\crefname{thm}{theorem}{theorems}

\newtheorem{corollary}{Corollary}[theorem]
\newenvironment{cor}{\begin{corollary}}{\hfill$\diamond$\end{corollary}}
\crefname{cor}{corollary}{corollaries}

\newtheorem{lemma}{Lemma}[section]
\newenvironment{lem}{\begin{lemma}}{\hfill$\diamond$\end{lemma}}
\crefname{lem}{lemma}{lemmas}

\newtheorem{definition}{Definition}[section]
\newenvironment{mydef}{\begin{definition}}{\hfill$\diamond$\end{definition}}
\crefname{mydef}{definition}{definitions}

\newtheorem{problem}{Problem}[section]

\crefname{prm}{problem}{problems}

\makeatletter
\DeclareRobustCommand{\bigDelta}{\mathop{\vphantom{\sum}\mathpalette\bigDelta@\relax}\slimits@}
\newcommand{\bigDelta@}[2]{\vcenter{\sbox\z@{$#1\sum$}\hbox{\resizebox{.9\dimexpr\ht\z@+\dp\z@}{!}{$\m@th\Delta$}}}}
\makeatother

\begin{document}

\preprint{AIP/123-QED}

\newcommand{\axel}[1]{{\color{green}Axel: #1}} 
\newcommand{\sharpP}{\#$\mathbb{P}$-Complete}
\newcommand{\bs}[1]{{\bf #1}}

\title[Counting single-qubit Clifford equivalent graph states is \#$\mathbb{P}$-Complete]{Counting single-qubit Clifford equivalent graph states is \#$\mathbb{P}$-Complete}
\author{Axel Dahlberg}
\author{Jonas Helsen}
\author{Stephanie Wehner}
\affiliation{QuTech, Delft University of Technology, and Kavli Institute of Nanoscience\\Delft, The Netherlands}

\date{\today}% It is always \today, today,

\begin{abstract}
    Graph states, which include for example Bell states, GHZ states and cluster states, form a well-known class of quantum states with applications ranging from quantum networks to error-correction.
    Deciding whether two graph states are equivalent up to single-qubit Clifford operations is known to be decidable in polynomial time and have been studied both in the context of producing certain required states in a quantum network but also in relation to stabilizer codes.
    The reason for the latter this is that single-qubit Clifford equivalent graph states exactly corresponds to \emph{equivalent} stabilizer codes.
    We here consider the computational complexity of, given a graph state $\ket{G}$, counting the number of graph states, single-qubit Clifford equivalent to $\ket{G}$.
    We show that this problem is \sharpP.
    To prove our main result we make use of the notion of isotropic systems in graph theory.
    We review the definition of isotropic systems and point out their strong relation to graph states.
    We believe that these isotropic systems can be useful beyond the results presented in this paper. 
    %Furthermore, we provide an explicit algorithm for constructing the orbit of graph states under single-qubit Cliffords, with runtime {\color{red}$\mathcal{O}(k\cdot\abs{V(G)}^2)$}, where $k$ is the size of the orbit and $\abs{V(G)}$ is the number of vertices in the graph $G$.

\end{abstract}

\maketitle

\section{Introduction}
Graph states form a well-studied class of quantum states and feature in many applications in quantum networks and quantum computers.
In a quantum network, graph states are a resource used by applications such as secret sharing~\cite{Markham2008secretsharing}, anonymous transfer~\cite{Christandl2005anonymous} and others.
In a quantum computer, graph states are the logical codewords of many quantum error-correcting codes~\cite{Gottesman2004thesis} and form a universal resource for measurement-based quantum computing~\cite{Briegel2001measurement}.
The action of single-qubit Clifford operations on graph states is well-understood and can be characterized completely in terms of operations called \emph{local complementations}, acting on the corresponding graph~\cite{Nest2003LCLC}. When faced with a class of objects and an action on them it is natural to consider the orbits induced by this action. The orbit of graph states under single-qubit Clifford operations can be studied by considering the orbits of simple graphs under local complementations, through the mapping mentioned above. 
The orbits of graph states have been studied in for example~\cite{Danielsen2006classification}. There the motivation came from quantum error correction: graph states can be mapped to stabilizer codes and moreover, the number of orbits for a given number of qubits is equal to the number of \emph{equivalent} stabilizer codes. This gives a method to count the number of inequivalent stabilizer codes on a fixed number of qubits.
In~\cite{Danielsen2006classification} the number of inequivalent stabilizer codes is computed for up to 12 qubits by counting the number of orbits of graphs under local complementations.

Furthermore, when studying entanglement measures and the equivalence of quantum states under local operations, the orbits of graphs states under single-qubit Cliffords is naturally an important question.
In the excellent survey on graph states ~\cite{Hein2006} it is stated that the computational complexity of generating the orbit of a given graph states is unknown.
%Here we show that (1) given a graph $G$, counting the number of graph states $k$ equivalent to $\ket{G}$ under single-qubit Clifford operations is \sharpP\ and (2) generating these graphs can be done in time\footnote{Note that $k$ scales, in general, exponentially in $\abs{V(G)}$.} {\color{red}$\mathcal{O}(k\cdot\abs{V(G)}^2)$}.
Here we show that given a graph $G$, counting the number of graph states equivalent to $\ket{G}$ under single-qubit Clifford operations, i.e. deciding the size of the orbit, is \sharpP.
\sharpP\ problems are of great interest in the field of quantum computing.
The reason being that the problem of boson-sampling~\cite{Aaronson2013boson}, efficiently solvable using a quantum computer, has very strong similarities with the \sharpP\ problem of computing the permanent of a matrix~\cite{Valiant1979permanent}.

\section*{Related work}
% The formal description of stabilizer states and stabilizer codes was developed in the PhD thesis of Gottesman~\cite{Gottesman2004thesis}.
% In~\cite{Raussendorf2001oneway} Raussendorf and Briegel showed that a certain class of graph states form a universal resource for measurement-based quantum computing.
The action of single-qubit Clifford operations on graphs states was characterized by Van den Nest et al. in~\cite{VandenNest2004graphical}, where it was shown that these operations acting on a graph state can be completely described by the action of local complementations on the corresponding graph.
Furthermore, in~\cite{VandenNest2004efficient}, Van den Nest et al. used this fact to extend the efficient algorithm by Bouchet for deciding equivalence of graphs under local complementations~\cite{Bouchet1991efficient} to an efficient algorithm for deciding equivalence of graph states under single-qubit Clifford operations.

If one also allow for single-qubit Pauli measurements and classical communication, the problem turns out to be equivalent~\cite{dahlberg2018courcelle} to the known graph theory problem of deciding if a graph is a vertex-minor~\cite{Courcelle2007vertexminor,Oum2005rankwidth} of another.
We have previously used this fact to show that deciding if a graph state $\ket{H}$ can be reached from another $\ket{G}$ using only single-qubit Clifford operations, single-qubit Pauli measurements and classical communication ($\mathrm{LC}+\mathrm{LPM}+\mathrm{CC}$) is NP-Complete, even if $\ket{H}$ is restricted to be (1) a GHZ-state on a fixed subset of the qubits of $\ket{G}$~\cite{dahlberg2018long}, (2) a GHZ-state on some subset of the qubits of $\ket{G}$~\cite{dahlberg2019iso} and (3) the tensor product of Bell pairs between fixed qubits~\cite{dahlberg2019bellvm}.

However, even if a problem is $\mathbb{NP}$-Complete, one can often find efficient algorithms for certain restrictions of the problem.
A general concept is that of fixed-parameter tractability, where an algorithm solving a hard problem is shown to have a runtime $\mathcal{O}(f(r)\cdot\text{poly}(n))$, where $f$ is some computable function, $r$ is some parameter of the input and $n$ is the size of the input.
For $\mathbb{NP}$-Complete problems, $f(r)$ is necessarily super-polynomial in $n$, unless $\mathbb{P}=\mathbb{NP}$.
Nonetheless, a fixed-parameter tractable problem can therefore be solved in polynomial time on inputs where the parameter $r$ is bounded.
An extremely powerful result in this context is that of Courcelle~\cite{Courcelle2011book}, which states that any graph problem, expressible in a certain rich logic (MS)\footnote{Monadic second-order logic}, can be solved in time $\mathcal{O}(f(\mathrm{rwd}(G))\cdot\abs{V(G)}^3)$, where $\mathrm{rwd}(G)$ is the rank-width~\cite{Oum2005rankwidth} of $G$ and $\abs{V(G)}$ is the number of vertices of $G$.
In~\cite{Courcelle2007vertexminor} Courcelle and Oum showed that the vertex-minor problem is expressible in MS and therefore that it is fixed-parameter tractability in the rank-width of the input graph.
It turns out that the rank-width of a graph $G$ equals one plus the \emph{Schmidt-rank width} the graph state $\ket{G}$~\cite{VandenNest2007schmidt}.
Using these results, we applied Courcelle's theorem to the problem of transforming graph states under $\mathrm{LC}+\mathrm{LPM}+\mathrm{CC}$ in~\cite{dahlberg2018courcelle} and thus showed that this problem is fixed-parameter tractable in the Schmidt-rank width of the input graph state.

In this paper we will focus on the computational complexity of counting the number of graph states equivalent to some graph state using only single-qubit Clifford operations.
We point out that, since the property of whether a graph is locally equivalent to another is also expressible in MS~\cite{Courcelle2007vertexminor,dahlberg2018courcelle}, Courcelle's machinery can also be applied to this problem.
In fact, Courcelle's theorem also holds for counting the number of satisfying solutions~\cite{Courcelle2011book}, which is what we are interested in here.
The details for how to apply Courcelle's theorem to the problem at hand, we leave for another paper.
Here, we instead show that the problem is \sharpP, and thus has no efficient algorithm in the general case, unless $\mathbb{P}=\mathbb{NP}$.

\subsection*{Overview}
In \cref{sec:graph_states} we recall how graph states and single-qubit Cliffords relate to graphs and local complementations.
In \cref{sec:isotropic} we review the graph theoretical notion of an isotropic system and relate this to stabilizer and graph states.
In \cref{sec:complexity} we review the complexity class \sharpP.
In \cref{sec:counting} we prove our main result that counting the number of graph states equivalent under single-qubit Cliffords is \sharpP.
%Finally in \cref{sec:gen_orbit} we construct an algorithm which generates the equivalent graph states to with a running time of {\color{red}$\mathcal{O}(k\cdot\abs{V(G)}^2)$}, where $k$ is the size of the orbit and $\abs{V(G)}$ is the number of vertices in the graph $G$.

\subsection*{Notation}
We use the following notation for sets of consecutive natural numbers.
\begin{equation}
	[n] = \{i\in\mathbb{Z}\;:\;0\leq i < n\}
\end{equation}
For a vertex $u$ in a graph $G=(V,E)$ we will denote the \emph{neighborhood}, i.e. the adjacent vertices as
\begin{equation}
    N_G(v)=\{u\in V:(u,v)\in E\}.
\end{equation}
Furthermore given a subset $X\subseteq V$ we use the following notation for the symmetric difference of the neighborhoods of the vertices in $X$
\begin{equation}
	N_G(X)=\bigDelta_{v\in X}N_G(v).
\end{equation}
Given a graph $G$, we denote by $\overline{G}$ the \emph{complementary} graph, i.e. the graph with vertex-set $V$ and edge-set
\begin{equation}
    \overline{E}=\{(u,v)\in V\times V: u\neq v \;\land\; (u, v)\notin E\}.
\end{equation}

The Pauli matrices are denoted as
\begin{align}
	&I=\begin{pmatrix}1 & 0\\0 & 1\end{pmatrix},\quad &X=\begin{pmatrix}0 & 1\\1 & 0\end{pmatrix}, \nonumber\\ &Y=\begin{pmatrix}0 & -\mathrm{i}\\\mathrm{i} & 0\end{pmatrix},\quad &Z=\begin{pmatrix}1 & 0\\0 & -1\end{pmatrix}.
\end{align}
The single-qubit Pauli group $\mathcal{P}_1$ consists of $\{\mathrm{i}^kI,\mathrm{i}^kX,\mathrm{i}^kY,\mathrm{i}^kZ\}$ for $k\in \mathbb{Z}_4$ together with matrix-multiplication.
Single-qubit unitaries that take elements of $\mathcal{P}_1$ to elements of $\mathcal{P}_1$ are called single-qubit Clifford operations and formally form the normalizer of $\mathcal{P}_1$.
The $n$-qubit Pauli group $\mathcal{P}_n$ is the $n$-fold single-qubit Pauli group whose elements are the tensor-products of elements of $\mathcal{P}_1$.

\section{Graph states}\label{sec:graph_states}
Here we review graph states and their properties under single-qubit Clifford operations.
We start with reviewing stabilizer states: a subset of graph states.

\subsection{Stabilizer states}
A stabilizer state $\ket{\mathcal{S}}$ on $n$ qubits is defined by its stabilizer group $\mathcal{S}$, which is a subgroup of the Pauli group $\mathcal{P}_n$~\cite{Gottesman2004thesis}.
The stabilizer state is defined to be a state such that it is an eigenstate of all elements of $\mathcal{S}$ with an eigenvalue of $+1$, i.e. $s\ket{\mathcal{S}}=\ket{\mathcal{S}}$ for $s\in \mathcal{S}$.
To avoid $\ket{\mathcal{S}}$ being a trivial zero state there are two requirements of $\mathcal{S}$, (1) $-I\notin \mathcal{S}$ and (2) all elements of $\mathcal{S}$ should commute\footnote{Elements of the Pauli group either commute or anti-commute.}.
Furthermore, for $\ket{\mathcal{S}}$ to be a unique state (up to a global phase), $\mathcal{S}$ needs to be of size $2^n$ and can therefore be described by $n$ independent generators.
As an example consider the stabilizer group $\mathcal{S}_0$ generated by $X\otimes X$ and $Z\otimes Z$.
One can check that $\mathcal{S}_0$ describes the state
\begin{equation}
	\ket{\mathcal{S}_0} = \frac{1}{\sqrt{2}}\left(\ket{0}\otimes\ket{0} + \ket{1}\otimes\ket{1}\right)
\end{equation}

\subsection{Graph states}
A graph state is a multi-partite quantum state $\ket{G}$ which is described by a graph $G$, where the vertices of $G$ correspond to the qubits of $\ket{G}$~\cite{Hein2006}.
The graph state is formed by initializing each qubit $v\in V(G)$ in the state $\ket{+}_v=\frac{1}{\sqrt{2}}(\ket{0}_v+\ket{1}_v)$ and for each edge $(u,v)\in E(G)$ applying a controlled phase gate between qubits $u$ and $v$.
Importantly, all the controlled phase gates commute and are invariant under changing the control- and target-qubits of the gate.
This allows the edges describing these gates to be unordered and undirected.
Formally, a graph state $\ket{G}$ is given as
\begin{equation}
    \ket{G}=\prod_{(u,v)\in E(G)}C_Z^{(u,v)}\left(\bigotimes_{v\in V(G)}\ket{+}_v\right),
\end{equation}
where $C_Z^{(u,v)}$ is a controlled phase gate between qubit $u$ and $v$, i.e.
\begin{equation}
    C_Z^{(u,v)}=\ket{0}\bra{0}_u\otimes\mathbb{I}_v+\ket{1}\bra{1}_u\otimes Z_v
\end{equation}
and $Z_v$ is the Pauli-$Z$ matrix acting on qubit $v$.
As an example, the graph state described by the complete graph on two vertices $K_2$ is single-qubit Clifford equivalent to each of the four Bell pairs since
\begin{equation}
    \ket{K_2} = \frac{1}{\sqrt{2}}(\ket{0}_a\otimes\ket{+}_b+\ket{1}_a\otimes\ket{-}_b) = H_b\ket{\Phi^+}_{ab}
\end{equation}
where $\ket{+}=(\ket{0}+\ket{1})/\sqrt{2}$, $H_b$ is a Hadamard gate on qubit $b$ and where
\begin{equation} 
    \ket{\Phi^+}_{ab} = \frac{1}{\sqrt{2}}(\ket{0}_a\otimes\ket{0}_b+\ket{1}_a\otimes\ket{1}_b).
\end{equation}

A graph state is also a stabilizer state~\cite{Hein2006} with a stabilizer group generated by
\begin{equation}\label{eq:graph_stab}
	g_v = X_v\prod_{u\in N_v} Z_u\quad\text{for}\quad v\in V(G)
\end{equation}
Furthermore, any stabilizer state is single-qubit Clifford equivalent to some graph state~\cite{VandenNest2004graphical}.

Importantly here is that the above in fact gives a bijective mapping from graphs to graph states.
Formally we have the following theorem.
\begin{lem}\label{lem:bijective}
    Two graphs states $\ket{G}$ and $\ket{G'}$ are equal if and only if their corresponding graphs $G$ and $G'$ are equal.
    %The mapping
    %\begin{equation}
    %    G\mapsto \lambda(G)=\ket{G}
    %\end{equation}
    %is one-to-one.
\end{lem}
\begin{proof}
    Let $\ket{G}$ and $\ket{G'}$ be two graph states.
    If $G$ and $G'$ have differing vertex-sets then clearly $\ket{G}$ and $\ket{G'}$ are different since they are states on different sets of qubits.
    Assume now that $G$ and $G'$ are graphs with the same vertex-set $V$.
    The inner product between $\ket{G}$ and $\ket{G'}$ will then be given as\footnote{We assume here that when iterating over a set, the order of the elements is always the same.}
    \begin{equation}
        \braket{G}{G'} = \left(\bigotimes_{v\in V}\bra{+}_v\right)\prod_{(u,v)\in E(G)}C_Z^{(u,v)}\prod_{(u,v)\in E(G')}C_Z^{(u,v)}\left(\bigotimes_{v\in V}\ket{+}_v\right).
    \end{equation}
    Using the fact that $C_Z^{(u,v)}$ commute and square to identity for any $(u,v)$ we find that the above equation evaluates to
    \begin{equation}
        \braket{G}{G'} = \left(\bigotimes_{v\in V}\bra{+}_v\right)\ket{G + G'}
    \end{equation}
    where $G + G'$ is the graph with vertex-set $V$ and edge-set $E(G)\Delta E(G')$ with $\Delta$ being the symmetric difference.
    The state $\ket{G + G'}$ is equal to $\bigotimes_{v\in V}\bra{+}_v$ if and only if $G + G'$ is the empty graph.
    One can see this by for example considering the Schmidt-rank for a bipartition which separates some adjacent vertices in $G + G'$, since this would be one for $\bigotimes_{v\in V}\ket{+}_v$ and greater than one for $\ket{G + G'}$.
    We therefore have that $\braket{G}{G'}$ is one if and only if $G$ and $G'$ are equal.
    %In~\cite{} it is shown that the inner product of any two graph states equals
    %\begin{equation}
    %    \bra{G'}\ket{G} = 1 - \abs{E(G)\Delta E(G'))}
    %\end{equation}
    %where $\Delta$ is the symmetric difference.
    %Clearly if $\ket{G'}=\ket{G}$ then $\abs{E(G)\Delta(G')}=1$ and $G=G'$.
\end{proof}

%It turns out that two graph states $\ket{G}$ and $\ket{G'}$ are single-qubit Clifford equivalent if and only if the graphs $G$ and $G'$ are equivalent under a operation called \emph{local complementation}~\cite{}.
It turns out that single-qubit Clifford operations on graph states can be completely captured by an operation called \emph{local complementation} on the corresponding graphs.

\begin{mydef}[Local complementation]\label{def:LC}
	A local complementation $\tau_v$ is a graph operation specified by a vertex $v$, taking a graph $G$ to $\tau_v(G)$ by replacing the induced subgraph on the neighborhood of $v$, i.e.  $G[N(v)]$, by its complement.
    The neighborhood of any vertex $u$ in the graph $\tau_v(G)$ is therefore given by
    % A local complementation $\tau_v$ is a graph operation, taking a graph $G$ to $\tau_v(G)$ with vertex-set $V(\tau_v(G))=V(G)$ such that the neighborhoods of $\tau_v(G)$ are given by
    % A local complementation $\tau_v$ acts on a vertex $v$ of a graph $G$ by complementing the induced subgraph on the neighborhood of $v$. The neighborhoods of the graph $\tau_v(G)$ are therefore given by
    \begin{equation}
	    N(u)^{(\tau_v(G))}=\begin{cases}N(u)\Delta (N(v)\setminus\{u\}) & \quad \text{if } (u,v)\in E(G) \\ N(u) & \quad \text{else}\end{cases},
    \end{equation}
    where $\Delta$ denotes the symmetric difference between two sets.
    %Given a sequence of vertices $\bs{v}=v_1\dots v_k$, we denote the induced sequence of local complementations, acting on a graph $G$, as
    %\begin{equation}
    %    \tau_{\bs{v}}(G)=\tau_{v_k}\circ\dots\circ\tau_{v_1}(G).
    %\end{equation}
\end{mydef}

The action of a local complementation on a graph induces the following sequence of single-qubit Clifford operations on the corresponding graph state
\begin{equation}\label{eq:LCU}
	U_v^{(G)}=\exp\left(-\mathrm{i}\frac{\pi}{4}X_v\right)\prod_{u\in N_v}\exp\left(\mathrm{i}\frac{\pi}{4}Z_u\right),
\end{equation}
where $X_v$ and $Z_v$ are the Pauli-X and Pauli-Z matrices acting on qubit $v$ respectively.
Concretely, $U_v^{(G)}$ has the following action on the graph state $\ket{G}$
\begin{equation}
    U_v^{(G)}\ket{G}=\ket{\tau_v(G)}.
\end{equation}

We call two graphs which are related by some sequence of local complementations \emph{locally equivalent}.
For example star graphs and complete graphs on the same vertex-set are locally equivalent.
As shown by Bouchet in~\cite{Bouchet1991efficient}, deciding if two graphs are locally equivalent can be done in cubic time in the number of vertices of the graphs.
The following theorem, proven by Van den Nest in~\cite{Nest2003LCLC}, captures the relation between single-qubit Cliffords on graph states and local complementations on graphs.

\begin{thm}[Van den Nest~\cite{Nest2003LCLC}]\label{thm:LCLC}
    Two graph states $\ket{G}$ and $\ket{G'}$ single-qubit Clifford equivalent if and only if the two graphs $G$ and $G'$ are locally equivalent.
\end{thm}

As a direct corollary of \cref{lem:bijective} and \cref{thm:LCLC}, we therefore have the following result.
\begin{cor}\label{cor:equal}
    Let $G$ be a graph with its corresponding graph state $\ket{G}$.
    The number of graph states which are single-qubit Clifford equivalent to $\ket{G}$ is equal to the number of locally equivalent graphs to $G$.
\end{cor}

Using \cref{cor:equal} we can now restrict ourselves to the problem of counting locally equivalent graphs.

\section{Isotropic systems}\label{sec:isotropic}
Our main result of this paper makes heavy use of the concept of an \emph{isotropic system}.
In this section we review the definition of an isotropic system and its relation to locally equivalent graphs and graph states.
What is interesting to point out, and perhaps never before noted, is that an isotropic system is in fact equivalent to a stabilizer group, see below.
For this reason, results obtained for isotropic systems can be of great use when studying stabilizer states and graph states.

Isotropic systems were introduced by Bouchet in~\cite{Bouchet1987isotropic}.
The power of isotropic systems is that they exactly capture the equivalence classes of graphs under local complementation or equivalently equivalence classes of graphs states under single-qubit Clifford operations.
Any isotropic system has a set of \emph{fundamental graphs} which are all locally equivalent.
As shown in~\cite{Bouchet1988graphicisotropic}, any graph $G$ is a fundamental graph of some isotropic system $S$.
Furthermore, given a isotropic system $S$ with a fundamental graph $G$, another graph $G'$ is a fundamental graph of $S$ if and only if $G$ and $G'$ are locally equivalent\footnote{That is, if and only if $\ket{G}$ and $\ket{G'}$ are equivalent under single-qubit Clifford operations.}.

In \cref{subsec:isotropic} we review the formal definition of an isotropic system.
In the sections leading up to this, we first set the notation and introduce certain concepts needed.

\subsection{Finite fields and Pauli groups}
Let $\{0, 1, \omega, \omega^2\}$ be the elements of the finite field of four elements $\mathbb{F}_4$.
Under addition we have that $x+x=0$ for any element $x$ of $\mathbb{F}_4$ and furthermore we have that $1+\omega=\omega^2$.
Under multiplication we have that $x^i\cdot x^j = x^{i+j\pmod{3}}$ for any element $x\neq 0$.
An useful inner product on $\mathbb{F}_4$ is the trace inner product, defined as
\begin{equation}\label{eq:trace}
	\langle a, b\rangle = a \cdot b^2 + a^2 \cdot b
\end{equation}
What is interesting in relation to quantum information theory is that addition in $\mathbb{F}_4$ corresponds to matrix multiplication in the Pauli group, up to a global phase.
Furthermore, the trace inner product captures whether two elements of the Pauli group commute or not.
To see this, consider the following mapping $\alpha$ from $\mathbb{F}_4$ to the Pauli group.
\begin{equation}\label{eq:alpha_multiply}
	\alpha(0)=I,\;\alpha(1)=X,\;\alpha(\omega)=Y,\;\alpha(\omega^2)=Z.
\end{equation}
One can then check that
\begin{equation}
	\alpha(a)\alpha(b) = \mathrm{i}^k\alpha(a + b)\quad\text{where}\quad k\in\mathbb{Z}_4.
\end{equation}
Furthermore, we have that
\begin{equation}\label{eq:alpha_commutate}
	[\alpha(a),\alpha(b)] = 0\quad\Leftrightarrow\quad \langle a,b\rangle=0.
\end{equation}
where $[\cdot,\cdot]$ is the commutator.
Similarly we can also define a map from the elements of the vector space $\mathbb{F}_4^n$.
Let $\mathbf{v}$ be a vector of $\mathbb{F}_4^n$ and $\mathbf{v}^i$ be the $i$'th element of $\mathbf{v}$.
We then define a map from $\mathbb{F}_4^n$ to the Pauli group on $n$ qubits as follows.
\begin{equation}\label{eq:alpha}
	\alpha(\mathbf{v})=\bigotimes_{i=0}^n\alpha(\mathbf{v}^i).
\end{equation}
Both \cref{eq:alpha_multiply} and \cref{eq:alpha_commutate} hold for $\alpha$ also acting on vectors of $\mathbb{F}_4^n$.

\subsection{Isotropic systems}\label{subsec:isotropic}
Formally an isotropic system is defined as follows\footnote{The definition here is equivalent to the original one in~\cite{Bouchet1987isotropic}, however we use a slightly different notation than Bouchet used 30 years ago.}.
\begin{mydef}[isotropic system]
	A subspace $S$ of $\mathbb{F}_4^n$ is said to be an isotropic system if: (1) for all $\mathbf{v},\mathbf{w}\in S$ it holds that $\langle \mathbf{v}, \mathbf{w}\rangle=0$ and (2) $S$ has dimension $n$.
\end{mydef}
Now note that $\alpha(S)\equiv\{\alpha(\mathbf{v})\;:\;\mathbf{v}\in S\}$ forms a stabilizer group (ignoring global phases).
This is because condition (1) of the above definition says that all the elements of $\alpha(S)$ commute, by \cref{eq:alpha_commutate}, as required by a stabilizer group.
% Furthermore, it is easy to see that any stabilizer group can be written as $\alpha(S)$ for some isotropic system $S$.

% \axel{Should we make these things theorems and formally prove them? This was not really my original plan with the paper but rather just to point these things out. Jonas: I don't think we should, its pretty obvious and not part of the point of the paper}

\subsection{Complete and Eulerian vectors}
Here we review some further concepts related to isotropic systems which we need for the proof of our main result.
Certain isotropic systems can be represented as Eulerian tours on $4$-regular multi-graphs (see \cref{sec:graphic}).
These Eulerian tours correspond to what are called Eulerian vectors of the isotropic system.
The definition of a Eulerian vector also generalizes to all isotropic systems, even those not representable as Eulerian vectors on $4$-regular multi-graphs.
In order to give the definition of an Eulerian vector we must first define what are called complete vectors. 

\begin{mydef}[complete vector]
	A vector $\mathbf{v}$ of $\mathbb{F}_4^n$ such that $\mathbf{v}^i\neq 0$ for all $i\in[n]$ is called \emph{complete}.
\end{mydef}

In the coming sections we will also need to notion of supplementary vectors.

\begin{mydef}[supplementary vectors]
	Two vectors $\mathbf{v},\mathbf{w}$ of $\mathbb{F}_4^n$ are called supplementary if (1) they are complete and (2) $\mathbf{v}^i\neq \mathbf{w}^i$ for all $i\in[n]$.
\end{mydef}
Complete vectors come equipped with a notion of \emph{rank}.
To define the rank of a complete vector we need some further notation.
Let $\mathbf{v}$ be a complete vector of $\mathbb{F}_4^n$.
Let $X$ be a subset of $[n]$ and let $\mathbf{v}[X]$ be a vector such its elements are
\begin{equation}
	(\mathbf{v}[X])^i = \begin{cases} \mathbf{v}^i & \text{if } i\in X \\ 0 & \text{else} \end{cases}
\end{equation}
We can now define the following set
\begin{equation}
	V_\mathbf{v} = \{\mathbf{v}[X]\;:\; X\subseteq [n]\}.
\end{equation}
Note that $V_\mathbf{v}$ forms a subspace of $\mathbb{F}_4^n$.
The rank of $\mathbf{v}$ with respect to $S$ is now defined as the dimension of the intersection of $V_\mathbf{v}$ and $S$.

\begin{mydef}[rank of a complete vector]
	Let $\mathbf{v}$ be a complete vector of $\mathbb{F}_4^n$.
	The rank of $\mathbf{v}$, $r_S(\mathbf{v})$, with respect to $S$ is the dimension of the intersection of $V_\mathbf{v}$ and $S$, i.e. 
	\begin{equation}
		r_S(\mathbf{v}) = \mathrm{dim}(V_\mathbf{v} \cap S)
	\end{equation}
\end{mydef}

We are now ready to formally define an Eulerian vector of an isotropic system.

\begin{mydef}[Eulerian vector]
	A complete vector $\mathbf{v}$ of $\mathbb{F}_4^n$, such that $r_S(\mathbf{v})=0$ is called an Eulerian vector of $S$.
\end{mydef}

\subsection{Fundamental graphs}\label{sec:fundamental}
As mentioned, the power of isotropic systems is that their \emph{fundamental graphs} are exactly the graphs in an equivalence class under local complementations.
Here we review the definition of fundamental graphs of an isotropic system, which is defined by a Eulerian vector through a \emph{graphic description}.
\begin{mydef}[graphic presentation]
	Let $G$ be a graph with vertices\footnote{Note that we can always choose such a labeling of the vertices of $G$.} $V(G)=[n]$ and $\mathbf{v},\mathbf{w}$ be supplementary vectors of $\mathbb{F}_4^n$.
	The following is then an isotropic system
	\begin{equation}\label{eq:graphic_presentation}
		S=\{\mathbf{v}[N_X] + \mathbf{w}[X]\;:\;X\subseteq V(G)\}.
	\end{equation}
	The tuple $(F,\mathbf{v},\mathbf{w})$ is called a graphic presentation of the isotropic system $S$.
	Furthermore, $F$ is called a fundamental graph of $S$.
\end{mydef}

Note that
\begin{equation}\label{eq:iso_basis}
	\{\mathbf{v}[N_v] + \mathbf{w}[\{v\}]\;:\;v\in V(G)\}
\end{equation}
forms a basis for $S$.
In~\cite{Bouchet1988graphicisotropic} it is shown that if $(G,\mathbf{v},\mathbf{w})$ is a graphic presentation of $S$ then $\mathbf{v}$ is a Eulerian vector of $S$.
Furthermore, it is shown that given an Eulerian vector $\mathbf{v}$ of $S$, there exists a unique graphic presentation $(G,\tilde{\mathbf{v}},\mathbf{w})$ of $S$ such that $\mathbf{v}=\tilde{\mathbf{v}}$.
Note that two Eulerian vectors can still represent the same fundamental graph.

The observant reader will notice the close similarity between \cref{eq:iso_basis} and \cref{eq:graph_stab}.
Indeed, consider the two supplementary vectors $\mathbf{v}_{\omega^2} = (\omega^2,\dots,\omega^2)$ and $\mathbf{w}_1=(1,\dots,1)$.
Now, let $G$ be an arbitrary graph and $S$ be the isotropic system with $(G,\mathbf{v}_{\omega^2}, \mathbf{w}_1)$ as a graph presentation, as by \cref{eq:graphic_presentation}.
We will here call $S_G$ the \emph{canonical} isotropic system of $G$.
We then have that $\alpha(S_G)$ is exactly the stabilizer group of the graph state $\ket{G}$.
To see this note that
\begin{equation}
	g_v = \alpha(\mathbf{v}[N_v] + \mathbf{w}[\{v\}])\quad\forall v\in V(G)
\end{equation}
using \cref{eq:alpha} and \cref{eq:graph_stab}.

As mentioned before, any two graphs $G$ and $G'$ are locally equivalent if and only if they are fundamental graphs of the same isotropic system~\cite{Bouchet1988graphicisotropic}.
Furthermore, any graph is a fundamental graph of some isotropic system.
We therefore see that, for any isotropic system $S$, there exists a surjective map from the set of Eulerian vectors of $S$ to the graphs in an equivalence class of graphs under local complementations.
As described in the next section, for certain isotropic systems, the number of Eulerian vectors equals the number of Eulerian tours on some $4$-regular multi-graph.
We will make use of this fact to prove our main result.

\subsection{Graphic systems}\label{sec:graphic}
Certain isotropic systems, called \emph{graphic} systems, can be represented as a $4$-regular multi-graphs.
There is then a surjective map from the Eulerian tours on the $4$-regular multi-graph to the fundamental graphs of the graphic system~\cite{Bouchet1988graphicisotropic}.
The set of fundamental graphs for graphic systems is exactly the set of circle graphs~\cite{Bouchet1988graphicisotropic}.
We will briefly describe this relation here, however leaving out some details which are out of scope for this paper.
For details on graphic systems see~\cite{Bouchet1988graphicisotropic}, for circle graphs see~\cite{Bouchet1994circle,Golumbic2004} and it's relation to graph states see~\cite{dahlberg2018long}.

A $4$-regular multi-graph $F$ is a multi-graph (i.e. can contain multi-edges and self-loops) where each vertex has degree 4, i.e. $\abs{N_F(v)}=4\;\forall v\in V(F)$.
A \emph{walk} $P$ on $F$ is an alternating sequence of vertices and edges
\begin{equation}
	P = v_1e_1v_2\dots e_{k}v_{k+1}
\end{equation}
such that $e_i$ is incident on $v_i$ and $v_{i+1}$ for $i\in[n]$.
A \emph{trail} is a walk with no repeated edges.
A \emph{tour} is a \emph{trail} such that $v_1=v_{k+1}$.
An Eulerian tour is a tour which traverses all edges of $F$.
A $4$-regular multi-graph has at least one Eulerian tour, since all vertices have even degree~\cite{Euler1736}.
Any Eulerian tour on a $4$-regular multi-graph $F$ traverses each vertex exactly twice, except for the vertex which is both the start and the end of the tour.
The order in which these vertices are traversed is captured by the \emph{induced double-occurrence word}.
%Such a Eulerian tour induces therefore a double-occurrence word, the letters of which are the vertices of $F$, and consequently a circle graph as described in the following definition.

\begin{mydef}[Induced double-occurrence word]\label{def:eul_tour}
    Let $F$ be a connected $4$-regular multi-graph on $k$ vertices $V(F)$.
    Let $U$ be a Eulerian tour on $F$ of the form
    \begin{equation}\label{eq:eul_tour}
        U=x_1e_1x_2\dots x_{2k-1}e_{2k-1}x_{2k}e_{2k}x_1.
    \end{equation}
    with $x_i\in V(F)$ and $e_i\in E(F)$.
    %Note that every element of $V$ occurs exactly twice in $U$, except $x_0$.
    From a Eulerian tour $U$ as in \cref{eq:eul_tour} we define an induced double-occurrence word as
    \begin{equation}
        m(U)=x_1x_2\dots x_{2k-1}x_{2k}.
    \end{equation}
    %To denote the alternance graph given by the double-occurrence word induced by a Eulerian tour, we will write $\mathcal{A}(U)\equiv\mathcal{A}(m(U))$.
\end{mydef}

We can now define a mapping from an induced double-occurrence word $m(U)$ to a graph $\mathcal{A}(m(U))$, where the edges of $\mathcal{A}(m(U))$ are exactly the pairs of vertices in $m(U)$ which alternate.
Formally we have the following definition.

\begin{mydef}[Alternance graph]\label{def:alt_graph}
    Let $m(U)$ be the induced double-occurrence word of some Eulerian tour $U$ on some $4$-regular multi-graph $F$.
    Let now $\mathcal{A}(m(U))$ be a graph with vertices $V(F)$ and the edges $E$, such that for all $(u,v)\in V(F)\otimes V(F)$, $(u,v)\in E$ if and only if $m(U)$ is of the form
    \begin{equation}
        \dots u \dots v \dots u \dots v \dots \quad\text{or}\quad \dots v \dots u \dots v \dots u \dots,
    \end{equation}
    i.e. $u$ and $v$ are alternating in $m(U)$.
    We will sometimes also write $\mathcal{A}(U)$ as short for $\mathcal{A}(m(U))$.
\end{mydef}

It turns out that the set of alternating graphs induced by the Eulerian tours on some $4$-regular multi-graph $F$ are exactly the fundamental graphs of some isotropic system $S$.
We then say that $S$ is \emph{associated} to $F$.
An isotropic system that is associated to some $4$-regular multi-graph is called \emph{graphic}.
There is a formal mapping $\lambda$ from a $4$-regular multi-graph $F$ together with an ordering $T$ of its edges to an isotropic system $S=\lambda_T(F)$.
However, this mapping is rather complex and the interested reader can find the details in~\cite{Bouchet1988graphicisotropic}.
What is important here is that, for any $T$, there is a bijective mapping from the Eulerian tours of $F$ to the Eulerian vectors of $S=\lambda_T(F)$~\cite{Bouchet1988graphicisotropic}.
This statement is implied by the results developed in~\cite{Bouchet1988graphicisotropic}, however in a non-trivial way.
For this reason, we here point out why this follows in the following section.

\subsection{Eulerian decompositions}
% For graphic systems there is a bijective mapping from the fundamental graphs of the isotropic system (or equivalently from the Eulerian vectors of the isotropic system) to the Eulerian tours on some $4$-regular multi-graph.
% This fact follows from results developed by Bouchet in~\cite{Bouchet1988graphicisotropic}, however in a non-trivial way.
% For this reason, we here point out why this follows.

% We start with reviewing the notion of $4$-regular multi-graphs and Eulerian tours on such.
% A $4$-regular multi-graph $F$ is a multi-graph (i.e. can contain multi-edges and self-loops) where each vertex has degree 4, i.e. $\abs{N_F(v)}=4\;\forall v\in V(F)$.
% A \emph{walk} $P$ on $F$ is an alternating sequence of vertices and edges
% \begin{equation}
% 	P = v_1e_1v_2\dots e_{k}v_{k+1}
% \end{equation}
% such that $e_i$ is incident on $v_i$ and $v_{i+1}$ for $i\in[n]$.
% A \emph{trail} is a walk with no repeated edges.
% A \emph{tour} is a \emph{trail} such that $v_1=v_{k+1}$.
% An Eulerian tour is a tour which traverses all edges of $F$.
% A $4$-regular multi-graph has at least one Eulerian tour, since all vertices have even degree~\cite{Euler1736}.

An Eulerian decomposition $D$ of a $4$-regular multi-graph $F$ is a set of tours on $F$ such that each edges of $F$ is in exactly one of the tours.
As shown in~\cite{Bouchet1988graphicisotropic}, given a $4$-regular multi-graph on $n$ vertices, any Eulerian decomposition can be describe by a complete vector of $\mathbb{F}_4^n$.
To see this, note that an Eulerian decomposition on a $4$-regular multi-graph $F$ can be described by, for each vertex $v$ in $F$, a pairing of the incident edges on $v$.
For example, let $e_v^1$, $e_v^2$, $e_v^3$ and $e_v^4$ be the four edges incident on the vertex $v$ and consider now a pairing where $e_v^1$ is paired with $e_v^2$ and $e_v^3$ with $e_v^4$, written as $((e_v^1, e_v^2), (e_v^3,e_v^4))$.
We can then construct an Eulerian decomposition by walking along the vertices and edges of $F$ and when we reach $v$ through the edge $e_v^1$ we should exit through the edge $e_v^2$ and vice versa.
Note that there are exactly three different ways to pair the four edges of a vertex and we can thus represent this pairing by a non-zero element of $\mathbb{F}_4$ as
\begin{align}
	1 &\mapsto ((e_v^1, e_v^2), (e_v^3,e_v^4)) \\ \omega &\mapsto ((e_v^1, e_v^3), (e_v^2,e_v^4)) \\ \omega^2 &\mapsto ((e_v^1, e_v^4), (e_v^2,e_v^3)).
\end{align}
Furthermore we can represent the pairings of all the vertices of $F$ as a complete vector of $\mathbb{F}_4^n$.
Note that the Eulerian decomposition for a given complete vector depends on the ordering of the edges incident on a vertex.
However this ordering simply changes which Eulerian decomposition is related to which complete vector, but not the fact that we now have a mapping from complete vectors of $\mathbb{F}_4^n$ to Eulerian decompositions of $F$.
This ordering $T$ is exactly the ordering mentioned in the previous section, which can be used to map $F$ to an isotropic system $S=\lambda_T(F)$.
% For this reason, let $T$ denote the specific ordering of the edges incident on a vertex, for all vertices in $F$.
Let now $D_T(\mathbf{v})$ be the Eulerian decomposition induced by the complete vector $\mathbf{v}$ as described above.

Importantly here, as stated in~\cite{Bouchet1988graphicisotropic}, is that, for any Eulerian decomposition $D$ of $F$ there is a unique complete vector $\mathbf{v}\in\mathbb{F}_4^n$ such that $D=D_T(\mathbf{v})$, for a fixed $T$.
% However, note that this statement is independent of $T$.
% In~\cite{Bouchet1988graphicisotropic} the notion of the ordering $T$ being associated with an isotropic system which fundamental graphs are circle graphs is defined.
% This definition is rather complex and not really necessary for the understanding of the results here.
% For this reason we leave it to the interested reader to look at~\cite{Bouchet1988graphicisotropic}.
% The Eulerian tours on $F$ induce circle graphs, for details see~\cite{dahlberg2018long} and the ordering $T$ associated with $S$ gives a relation between these induced circle graphs and the fundamental graphs of $S$.
% Concretely, there exists an associated ordering $T$ with respect to $S$ if and only if an alternance graph induced by a Eulerian tour on $F$ is a fundamental graph of $S$.
Furthermore, the Eulerian decomposition $D_T(\mathbf{v})$ consists of an Eulerian tour if and only if $\mathbf{v}$ is an Eulerian vector of $S=\lambda_T(F)$.
We therefore have the following corollary.
\begin{cor}[Implied by~\cite{Bouchet1988graphicisotropic}]\label{cor:main}
	Let $F$ be a $4$-regular multi-graph with $n$ vertices.
	Let $T$ be an ordering of its vertices as described above and formally defined in~\cite{Bouchet1988graphicisotropic}.
	% Let $S$ be an isotropic system which fundamental graphs are circle graphs such that there exists an associated ordering $T$ (such $S$ always exists~\cite{Bouchet1988graphicisotropic}).
	The number of Eulerian tours on $F$ equals the number of Eulerian vectors of $S=\lambda_T(F)$.
\end{cor}
\begin{proof}
	From above we already know that a Eulerian decomposition of $F$ is described by exactly one complete vector of $\mathbf{F}_4^n$ through the mapping $D_T$.
	Furthermore, the Eulerian decomposition $D_T(\mathbf{v})$ consists of exactly one Eulerian tour if and only if $\mathbf{v}$ is a Eulerian vector of $S=\lambda_T(F)$.
	Finally the number of Eulerian decompositions of $F$ that consists of exactly one Eulerian tour are clearly equal to the number of Eulerian tours on $F$.
\end{proof}

\subsection{Number of locally equivalent graphs}\label{sec:number_loc}
In~\cite{Bouchet1993} Bouchet showed that $l(G)$, the number of graphs locally equivalent to some graph $G$, is given by
\begin{equation}\label{eq:main_eq0}
    l(G)=\frac{e(S)}{k(S)}
\end{equation}
where $S$ is an isotropic system with $G$ as a fundamental graph and $e(S)$ is the number of Eulerian vectors of $S$ and $k(S)$ is an index of $S$.
We also have that if $S$ and $S'$ are isotropic systems which both have $G$ as a fundamental graph, then $e(S)=e(S')$ and $k(S)=k(S')$.
Using the canonical isotropic system we introduced in \cref{sec:fundamental} we can therefore also define
\begin{equation}\label{eq:eGeS}
	e(G)\equiv e(S_G),\quad k(G)\equiv k(S_G),
\end{equation}
such that
\begin{equation}\label{eq:main_eq}
    l(G)=\frac{e(G)}{k(G)}.
\end{equation}

% where $e(G)$ is the number of Eulerian vectors of a isotropic system that has $G$ as a fundamental graph and $k(G)$ is an index of the same isotropic system.
% It is consistent to call $e$ and $k$ properties of $G$ since, if $S$ and $S'$ are isotropic systems which both have $G$ as a fundamental graph, then $e(S)=e(S')$ and $k(S)=k(S')$.
%Furthermore, since $e$ and $k$ are actually properties of the isotropic system, i.e. the equivalence class under local complementations containing $G$, both $e$ and $k$ are invariant under local complementations when evaluated on a graph.

Below, we review the definition of $k(G)$ as presented in~\cite{Bouchet1993}.
The index $k(G)$ of a graph is given as
\begin{equation}
    k(G)=\begin{cases} \abs{\nu(G)^\bot}+2 & \quad \text{if } G \text{ is in the class } \mu \\ \abs{\nu(G)^\bot} & \quad \text{else} \end{cases}
\end{equation}
where the bineighborhood space $\nu(G)$ and the graph class $\mu$ are defined below and $^\bot$ denotes the orthogonal complement.
Firstly, we introduce the following notation that will help simplify some later expressions.

\begin{mydef}\label{def:binvec}
    Let $S=\{s_1,\dots,s_k\}$ be a set and $P\subseteq S$ a subset of $S$. We will associate to $P$ a binary vector $\overrightarrow{P}$ of length $k$ as follows:
    \begin{equation}
        \overrightarrow{P}^{(i)}=\begin{cases} 1 \quad \text{if } s_i\in P \\ 0 \quad \text{else} \end{cases}
    \end{equation}
    where $\overrightarrow{P}^{(i)}$ is the $i$-th element of $\overrightarrow{P}$.
    We denote the number of nonzero elements of $\overrightarrow{P}$ as $\abs{\overrightarrow{P}}$, such that $\abs{\overrightarrow{P}}=\abs{P}$.
    % We will also sometimes write $\overrightarrow{P_1\cap P_2}$ as $\overrightarrow{P_1}\cap\overrightarrow{P_2}$.
\end{mydef}

The base-set $S$ will here be the vertices $V$ of a graph $G$ and from the context it will always be clear which graph.
We will also use $\cdot$ to denote the element-wise product between two binary vectors, such that
\begin{equation}
	\overrightarrow{P_1\cap P_2} = \overrightarrow{P_1}\cdot\overrightarrow{P_2}.
\end{equation}

To define the graph class $\mu$ we first need to review the notion of a bineighborhood space.

\begin{mydef}[bineighborhood space]
    Let $G=(V,E)$ be a simple graph and $\overline{G}=(V,\overline{E})$ the complementary graph of $G$.
    For any $u,v\in V$ let
    \begin{equation}
        \nu_G(e)=\overrightarrow{N_G(u)\cap N_G(v)}.
    \end{equation}
    For any subset $E'\subseteq E\cup \overline{E}$, let 
    \begin{equation}
        \nu_G(E')=\sum_{e\in E'}\nu_G(e).
    \end{equation}
    We will sometimes write $\nu(e)$ or $\nu(E')$ if it is clear which graph is considered.
    A subset $C\subseteq E$ such that the number of edges in $C$ incident to any vertex in $G$ is even is called a cycle.
    We denote the set of cycles of $G$ as $\mathcal{C}(G)$.
    Let $\mathfrak{V}=\mathbb{Z}^{\abs{V}}_2$ be the binary vector space of dimensions $\abs{V}$ and consider the two subspaces
    \begin{equation}
        \overline{\mathfrak{E}}=\{\nu(E'):E'\subseteq \overline{E}\},\quad \mathfrak{C}=\{\nu(C):C\subseteq \mathcal{C}(G)\}
    \end{equation}
    The bineighborhood space $\nu(G)$ is defined as the sum of the two subspaces $\overline{\mathfrak{E}}$ and $\mathfrak{C}$, i.e.
    \begin{equation}
        \nu(G)=\overline{\mathfrak{E}}+\mathfrak{C}.
    \end{equation}
    %The dimension of the bineighborhood space is invariant under local complementations, as shown in~\cite{Bouchet1993}.
\end{mydef}

Finally the graph class $\mu$ is defined as follows.

\begin{mydef}[graph class $\mu$]\label{def:mu}
    A simple graph $G=(V,E)$ is said to be in the class $\mu$ if:
    \begin{enumerate}
        \item $d_G(v)=1\pmod{2}$ for every vertex $v\in V$. I.e. all vertices in $G$ should have an odd degree.
	\item $\abs{\nu(e)}=0\pmod{2}$ for all edges $e\in \overline{E}$. I.e. for every edge (u,v), not in $G$, the symmetric difference of the neighborhoods of $u$ and $v$ should have an even size.
	\item $\abs{\nu(C)}=\abs{C}\pmod{2}$ for all cycles $C\in\mathcal{C}(G)$. I.e., for all cycles $C$ of $G$, the number of non-zero elements of the $\nu(C)$ and the number of edges of $C$ should both be even or both be odd.
    \end{enumerate}
\end{mydef}

\section{Complexity}\label{sec:complexity}
The problems in $\mathbb{NP}$ are decision problems where YES-instances to the problem have proofs that can be checked in polynomial time.
For example the SAT-problem is in $\mathbb{NP}$, where one is asked to decide if a given boolean formula has a satisfying assignment of its variables~\cite{Cook1971SAT}.
On the other hand, problems where the NO-instances have proofs that can be checked in polynomial time are the problems in co-$\mathbb{NP}$.
A problem is said the be $\mathbb{NP}$-Complete if (1) it is in $\mathbb{NP}$ and (2) any other problem in $\mathbb{NP}$ can be reduced to this problem in polynomial time.
$\mathbb{NP}$-Complete problems are therefore informally the hardest problem in $\mathbb{NP}$.

\#$\mathbb{P}$ problems are the \emph{counting} versions of the $\mathbb{NP}$ problems.
For example, the \emph{counting} version of SAT (\#SAT) is to compute how many satisfying assignments a given boolean formula has.
\sharpP\ problems are the problems in \#$\mathbb{P}$ for which any other problem in \#$\mathbb{P}$ can be polynomially reduced to.
For example \#SAT is \sharpP~\cite{Valiant1979permanent}.
Note that \sharpP\ is at least as hard as $\mathbb{NP}$-Complete, since if we know the number of satisfying assignments we know if at least one exists.
Other well-known problems \sharpP\ are for example computing the permanent of a given boolean matrix or finding how many perfect matchings a given bipartite graph has~\cite{Valiant1979permanent}.

Recently, \sharpP\ problems have been the interest of the quantum computing community due to the problem of boson sampling~\cite{Aaronson2013boson}.
The boson sampling problem can be solved efficiently on a quantum computer.
Furthermore, the boson sampling problem can be related to the problem of estimating the permanent of a complex matrix.
Since computing the permanent is in general a \sharpP\ problem and thus believed the be infeasible to solve efficiently on a classical computer, the boson sampling problem is therefore is a strong candidate for a problem showing 'quantum supremacy'.
% Furthermore, the problem of computing the permanent of certain complex matrices can be reduced to the boson sampling.
% However, it is not yet know whether the restriction on these matrices are such that the problem stays \sharpP\ or not.

\section{Counting the number of locally equivalent graphs is \#\texorpdfstring{$\mathbb{P}$}{P}-Complete}\label{sec:counting}
Here we show our following main result.

\begin{thm}[main]\label{thm:main}
	Counting the number, $l(G)$,  of locally equivalent graphs to a given graph $G$ is \sharpP.
\end{thm}

We do this by showing that counting the number of Eulerian tours of a 4-regular multi-graph can be reduced in polynomial time to computing $l(G)$, where $G$ is a circle graph.
Since counting the number of Eulerian tours of a 4-regular multi-graph is \#$\mathbb{P}$-Complete~\cite{Ge2010}, the result follows.
By \cref{cor:equal} we have the following corollary.
\begin{cor}
	Counting the number of graph states which are single-qubit Clifford equivalent to a given graph state $\ket{G}$ is \sharpP.
\end{cor}
\begin{proof}
	Directly implied by \cref{thm:main} and \cref{cor:equal}.
\end{proof}

\subsection{Reducing \# of Eulerian tours to \# of local equivalent graphs}
Here we show how the problem of computing the number of Eulerian tours on a 4-regular multi-graph can be reduced in polynomial time to the problem of computing the number of locally equivalent graphs to some circle graph and thus provide the proof for \cref{thm:main}.
\begin{proof}[Proof of \cref{thm:main}]
	From \cref{cor:main} we know that for any $4$-regular multi-graph $F$, there exists an isotropic system $S=\lambda_T(F)$ such that the number of Eulerian vectors $e(S)$ equals the number of Eulerian tours on $F$.
	Let now $G$ be a fundamental graph of $S$.
	We then have that $e(G)=e(S)$, by \cref{eq:eGeS} and~\cite{Bouchet1988graphicisotropic}.
	% We know that for any isotropic system $S'$ with $G$ as a fundamental graph, $e(S')=e(S)$~\cite{Bouchet1988graphicisotropic}.
	% We then have that $e(G)=e(S)$, see \cref{sec:number_loc}.
	Furthermore, recall the $G$ is necessarily also an alternance graph induced by some Eulerian tour on $F$, see \cref{sec:graphic}.
	We can therefore compute the number of Eulerian tours on $F$ by computing $l(G)\cdot k(G)$, as by \cref{eq:main_eq}.
	As we show below, we can both find $G$ and compute $k(G)$ in polynomial time from which the theorem follows.
	% We can therefore find a basis for an isotropic system $S'$, such that $e(S')$ equals the number of Eulerian tours on $F$, in polynomial time, as follows
	% Thus if 
	% Any other isotropic system $S'$ with the same 
	% Furthermore, from \cref{eq:main_eq} we know that the number of locally equivalent graphs is related to the number of Eulerian vectors $e(S)$ of the corresponding isotropic system $S$.
	% Given $F$ we can find a basis for $S=\lambda_T(F)$ in polynomial time by the following steps:

	We can find $G$ in polynomial time as follows.
	\begin{enumerate}
		\item Find an Eulerian tour $U$ on $F$, can be done in polynomial time by Fleury's algorithm~\cite{Fleury1883}.
		\item Construct the alternance graph $G=\mathcal{A}(U)$ induced by $U$, can be done in polynomial time, see~\cite{dahlberg2018long}.
		% \item Let $S_G$ be the canonical isotropic system of $G$, see \cref{sec:fundamental}.
	\end{enumerate}

    In the rest of this section we show that $k(G)$ can be computed in polynomial time from which the main result follows.%, by \cref{eq:main_eq}.
    We will start by showing that determining if a graph $G$ is in the class $\mu$, see \cref{def:mu}, can be done in time $\mathcal{O}(\abs{V}^5)$.
    Note that there might be even faster ways to compute this, but we are here only interested to show that this can be done in polynomial time.
    We assume that the graph $G=(V,E)$ is represented by its adjacency matrix.
    To check if $G$ is in the class $\mu$ one needs to check the three properties in \cref{def:mu}:
    \begin{enumerate}
        \item Checking if all vertices have odd degree can be done in $\mathcal{O}(\abs{V}^2)$ time.
        \item Checking if $\abs{\nu(e)}$ is even for all edges can be done in $\mathcal{O}(\abs{V}^3)$ time, since there are $\mathcal{O}(\abs{V}^2)$ edges and computing $\abs{\nu(e)}$ can be done in linear time\footnote{By taking the inner product of the corresponding rows in the adjacency matrix.}.
        \item For the last property is not directly clear whether this can be done in polynomial time since we need to a priori check the property $\abs{\nu(C)}=\abs{C}\pmod{2}$ for all cycles in $G$, which might be exponentially many.
            As we will now show, we only need to check the property for the cycles in a cycle basis of $G$.
	    A cycle basis $\mathcal{CB}=\{C_1,\dots,C_k\}$, where $k=\mathcal{O}(\abs{V}^2)$, is a set of cycles such that any cycle of $G$ can be written as the symmetric difference of the elements of a subset of $\mathcal{CB}$.
            As shown in~\cite{Paton1969} a cycle basis of an undirected graph can be found in $\mathcal{O}(\abs{V}^2)$ time.
            Thus any cycle of $G$ can be written as
            \begin{equation}
                \bigDelta_{C'\in\mathcal{C}}C' 
            \end{equation}
            where $\mathcal{C}$ is a subset of the cycle basis $\mathcal{CB}$.
            % Representing the cycles by binary vectors as in \cref{def:binvec} we can write any cycle of $G$ as
            % \begin{equation}
            %     \sum_{i=1}^k x_i \overline{C_i}
            % \end{equation}
            % where $\{x_i\}_i$ are elements of $\mathbb{F}_2$.
            % We will by slight abuse of notation from now on sometimes write $\nu(\overline{C})=\nu(C)$.
            We then have that
            \begin{equation}
                \nu\left(\bigDelta_{C'\in\mathcal{C}}C'\right)=\sum_{C'\in\mathcal{C}}\nu(C').
            \end{equation}
            Thus we need to show that for any $\mathcal{CB}$
            \begin{equation}\label{eq:basis_prop}
                \abs{\sum_{C'\in\mathcal{C}}\nu(C')}=\abs{\bigDelta_{C'\in\mathcal{C}}C'}\pmod{2}\quad\forall \mathcal{C}\subseteq\mathcal{CB}
            \end{equation}
            if and only if
            \begin{equation}\label{eq:basis_prop2}
                \abs{\nu(C)}=\abs{C}\pmod{2}\quad\forall C\in \mathcal{CB}.
            \end{equation}
	    Lets first show that \cref{eq:basis_prop} implies \cref{eq:basis_prop2}.
	    \Cref{eq:basis_prop} states that the equation holds for every subset $\mathcal{C}$ of the elements of the cycle basis $\mathcal{CB}$.
	    In particular it should hold for the singletons $\mathcal{C}=\{C\}$, where $C\in\mathcal{CB}$.
	    Note that this directly implies \cref{eq:basis_prop2}.
	    For the rest of this section we now prove that \cref{eq:basis_prop2} implies \cref{eq:basis_prop}.
            % It is clear that if $\abs{\nu{(C)}}\neq\abs{C}\pmod{2}$ for some $C\in\mathcal{CB}$ then the expression in \cref{eq:basis_prop} is not true.
            % Which can be seen by choosing $\mathcal{C}=\{C\}$.
            % What we therefore need to show is that if \cref{eq:basis_prop2} is true then so is \cref{eq:basis_prop}.
            We will do this by induction on the size of $\mathcal{C}$.
            This is obviously true if $\abs{\mathcal{C}}=1$.
            Lets therefore assume that the statement is true for $\abs{\mathcal{C}}\leq k$ which we will show implies that it is also true for $\abs{\mathcal{C}}=k+1$.
            Lets assume that $\mathcal{C}$ is a subset of $\mathcal{CB}$ of size $k+1$ and that $\tilde{C}$ is an element of $\mathcal{C}$.
            Lets then consider the left-hand side of \cref{eq:basis_prop}
            \begin{equation}\label{eq:deriv1}
                \abs{\sum_{C'\in\mathcal{C}}\nu(C')}=\abs{\sum_{C'\in\mathcal{C}\setminus\{\tilde{C}\}}\nu(C')+\nu(\tilde{C})}.
            \end{equation}
            We will now make use of the fact that the size of the symmetric difference of two sets $S_1$ and $S_2$ is $\abs{S_1\Delta S_2}=\abs{S_1}+\abs{S_2}-2\abs{S_1\cap S_2}$.
	    Expressed in terms of binary vectors this relation reads $\abs{\overrightarrow{S_1} + \overrightarrow{S_2}} = \abs{\overrightarrow{S_1}} + \abs{\overrightarrow{S_2}} - 2\abs{\overrightarrow{S_1}\cdot\overrightarrow{S_2}}$.
	    We therefore have that \cref{eq:deriv1} evaluates to
            \begin{widetext}
            \begin{equation}
               \abs{\sum_{C'\in\mathcal{C}\setminus\{\tilde{C}\}}\nu(C')+\nu(\tilde{C})}=\abs{\sum_{C'\in\mathcal{C}\setminus\{\tilde{C}\}}\nu(C')}+\abs{\nu(\tilde{C})}-2\abs{\left(\sum_{C'\in\mathcal{C}\setminus\{\tilde{C}\}}\nu(C')\right)\cdot\nu(\tilde{C})}
            \end{equation}
            \end{widetext}
            The result then follows since when taking $\pmod{2}$, the last term in the above expression vanishes and the two first evaluate to
            \begin{equation}\label{eq:deriv2}
                \abs{\bigDelta_{C'\in\mathcal{C}\setminus\{\tilde{C}\}}C'}+\abs{\tilde{C}}
            \end{equation}
            where we used the induction hypothesis.
            By a similar argument one can see that the expression in \cref{eq:deriv2} equals $\pmod{2}$
            \begin{equation}
                \abs{\bigDelta_{C'\in\mathcal{C}}C'}.
            \end{equation}

            Thus the total time to check property 3 in \cref{def:mu} is $\mathcal{O}(\abs{V}^5)$.
            To see this, note that we need to check $\abs{\nu(C)}=\abs{C}\pmod{2}$ for all $C\in\mathcal{CB}$, which contains $\mathcal{O}(V^2)$ elements.
            To compute $\abs{\nu(C)}$, we compute $\nu(e)$, in linear time, for each of the $\mathcal{O}(V^2)$ elements of $C$, and add these together, also in linear time.
            %since there are $\mathcal{O}(\abs{V}^2)$ elements of $\mathcal{CB}$ and it takes $\mathcal{O}(\abs{V}^4)$ to compute $\nu(C)$ since there are $\mathcal{O}(\abs{V}^2)$ elements in $C$.
    \end{enumerate}

    Additionally to deciding if the graph $G$ is in the class $\mu$, we also need to compute $\abs{\nu{(G)}^\bot}$ to determine $k(G)$.
    This can be done by first finding bases for the subspaces $\overline{\mathfrak{E}}$ and $\mathfrak{C}$.
    For $\overline{\mathfrak{E}}$ a basis can be found as $\left\{\overrightarrow{\{e\}}:e\in \overline{E}\right\}$.
    As stated above we can also find a basis for $\mathfrak{C}$, i.e. the cycle basis, in $\mathcal{O}(\abs{V}^2)$ time.
    From the bases for $\overline{\mathfrak{E}}$ and $\mathfrak{C}$ we can find a basis for $\nu(G)$ in $\mathcal{O}(\abs{V}^3)$ time, by Gaussian elimination.
    The number of basis vectors we found for $\nu(G)$ is then the dimension of $\nu(G)$.
    From the dimension of $\nu(G)$ we can find the dimension of $\nu(G)^\bot$ as
    \begin{equation}
        \mathrm{dim}(\nu(G)^\top)=\abs{V}-\mathrm{dim}(\nu(G))
    \end{equation}
    and finally the size of $\nu(G)^\bot$ as
    \begin{equation}
        \abs{\nu(G)^\bot}=2^{\mathrm{dim}(\nu(G)^\bot)}.
    \end{equation}

    Thus there exist an algorithm to compute $k(G)$ with running time $\mathcal{O}(\abs{V}^5)$.
    This then implies that computing the number of Eulerian tours in a 4-regular multi-graph can be reduced in polynomial time to computing the number of locally equivalent graphs to some circle graph, by using \cref{eq:main_eq}, and therefore \cref{thm:main}.
\end{proof}

\section{Conclusion}
We have shown that counting the number of graph states equivalent under single-qubit Clifford operations is \sharpP.
To do this we have made heavy use of certain concepts in graph theory, mainly developed by Bouchet.
As it turns out these concepts, for example isotropic systems, are highly relevant for the study of stabilizer and graph states.
We hope that this paper can serve as not only a proof of our main theorem but also as a reference for those in quantum information theory interested in finding use for these graph theory concepts in their research.

\begin{acknowledgments}
AD, JH and SW were supported by an ERC Starting grant, an NWO VIDI grant, and the Zwaartekracht QSC.
\end{acknowledgments}

\bibliography{refs}
\end{document}